\documentclass[11pt]{article}

\usepackage{amsmath,amssymb,amsfonts }
\usepackage{amsthm}
\usepackage[backend=biber,style=numeric-comp,sorting=none]{biblatex}
\usepackage{ dsfont }
\usepackage[utf8]{inputenc}
\usepackage[T1]{fontenc}

\usepackage[margin=1in]{geometry} 
\usepackage{titling}
\usepackage{xcolor}
\usepackage{hyperref}

\newcommand{\ket}[1]{\left| #1 \right\rangle}
\newcommand{\bra}[1]{\left\langle #1 \right|}

\newcommand{\ketbra}[2]{\left| #1 \middle\rangle \!\middle\langle #2 \right|}
\newcommand{\id}{\mathbb{I}}
\newcommand{\Tr}{\mathrm{Tr}}

\usepackage{needspace}
\hypersetup{
    colorlinks=false,
    pdfborder={0 0 1.5},
    linkbordercolor={1 0 0},   
    citebordercolor={0 1 0},   
    urlbordercolor={0 0 1}, 
    pdftitle    = {Comment on "Quantum theory based on real numbers cannot be experimentally falsified": On the compatibility of physical principles with information theory for fermions},
    pdfauthor   = {Fatemeh Moradi Kalarde; Xiangling Xu; Marc-Olivier Renou},
    pdfsubject  = {Qubit quantum information theory (Q-QIT); fermionic quantum information theory (F-QIT); real quantum theory},
    pdfkeywords = {qubit quantum information theory, Q-QIT, fermionic quantum information theory, F-QIT, real quantum theory, operational independence, independent preparation, composite systems, superselection rules, local tomography, generalised probabilistic theories},
    pdflang     = {en},
    pdfcreator  = {LaTeX with hyperref},
    pdfdisplaydoctitle = true
}

\title{\textbf{Comment on ``Quantum theory based on real numbers cannot be experimentally falsified'': On the compatibility of physical principles with information theory for fermions}}
\author{Fatemeh Moradi Kalarde$^{1}$, Xiangling Xu$^{1}$ and Marc-Olivier Renou$^{1}$\\
\footnotesize{$^{1}$Inria Paris-Saclay, Bâtiment Alan Turing,}
\footnotesize{1 rue Honoré d'Estienne d'Orves, 91120 Palaiseau, France}}

\date{}

\newtheorem{postulate}{Postulate}
\newtheorem{proposition}{Proposition}
\newtheorem*{claim*}{Claim}

\theoremstyle{definition}
\newtheorem{definition}{Definition}
\theoremstyle{plain}

\begin{document}

\maketitle

\vspace{-1cm}
\begin{abstract}
The manuscript \emph{arXiv:2603.19208}~\cite{arXiv2603} proposes a physically motivated postulate to select the appropriate formulation of quantum theory over real Hilbert spaces, ruling out the theory considered in \emph{Nature 600, 625-629 (2021)}~\cite{Renou} in favour of an alternative theory which reproduces the predictions of standard complex quantum information theory. 
Here, we argue that, to qualify as a general physical principle, this postulate should also hold in fermionic quantum information theory, the standard framework describing information encoded in the presence or absence of identical fermions. We show, however, that it fails in this framework.
More broadly, our result highlights the importance of confronting proposed foundational principles with fermionic quantum information theory, a point that also deserves further examination in recent related works such as \emph{arXiv:2503.17307}~\cite{arXiv2503} and \emph{arXiv:2504.02808}~\cite{arXiv2504}.
\end{abstract}

\section{Introduction}

We refer to standard quantum information theory, formulated over complex Hilbert spaces, as \emph{qubit quantum information theory} (Q-QIT); see Appendix~\ref{Sec:InfoTheories} for a description of this framework. The search for formulations over real Hilbert spaces that reproduce the predictions of Q-QIT has attracted renewed interest in recent years~\cite{arXiv2603,Renou,arXiv2503,arXiv2504,Li2022}. A central issue in constructing such formulations is how to represent composite systems and independent preparations. For the present discussion, we distinguish the following two formulations:

\begin{itemize}
    \item Theory $T_1^\mathds{R}$, which is obtained from the Q-QIT formalism based on complex Hilbert spaces, as given in standard textbooks such as~\cite{NielsenChuang}, but restricted to $\mathds{R}$-Hilbert spaces. It is given by~\cite[Eq.~(2)]{arXiv2603} (equivalently,~\cite[Def.~2, Eq.~(7)]{arXiv2603} with $\otimes$ the Kronecker product): the set of independently preparable states is identified with the set of \emph{Kronecker tensor product states}. More precisely, in $T_1^\mathds{R}$, the set of states $\rho_{AB}$ that two distant experimentalists $A, B$ can locally and independently prepare without communication or shared resources, is the set of all Kronecker tensor product states $\rho_A\otimes\rho_B$. It was shown in~\cite{Renou} that $T_1^\mathds{R}$ is experimentally falsifiable as it does not recover all the predictions of Q-QIT\footnote{$T_1^\mathds{R}$ is called Real Quantum Theory (RQT) in~\cite{Renou}. Ref.~\cite{arXiv2603} uses the name RQT for both $T_1^\mathds{R}$ and $T_2^\mathds{R}$, which differ only in the independence assumption.}.
    \item Theory $T_2^{\mathds{R}}$, in which the set of independently preparable states is identified with the set of \emph{operationally independent states} (Definition~\ref{Def:OpIndep} below), as in~\cite[Def.~2, Eq.~(8)]{arXiv2603}. Concretely, this theory is constructed as a modification of $T_1^{\mathds{R}}$, in which the set of bipartite states regarded as independently preparable is enlarged relative to that of $T_1^{\mathds{R}}$ (see \cite[Prop.~3]{arXiv2603})\footnote{Ref.~\cite{arXiv2603} also uses a different real theory, called Real-Number Quantum Theory (RNQT)~\cite[Def.~1]{arXiv2603} and introduced in~\cite{arXiv2504}. RNQT modifies the composition rule and is isomorphic to Q-QIT; Ref.~\cite{arXiv2603} uses it in the proof, by showing that RNQT models are contained in $T_2^{\mathds{R}}$ models. Other related real formulations include the framework of~\cite{arXiv2503}, based on the postulate that we briefly discuss in Appendix~\ref{App:OtherPostulates}, and the approach based on real K\"ahler spaces~\cite{volovich2025real,aref2025notes}.}.
\end{itemize}

Thus, $T_2^{\mathds{R}}$ retains the real state and effect spaces of $T_1^{\mathds{R}}$ and only enlarges the set designated as independently preparable; this modification allows it to reproduce the predictions of Q-QIT~\cite[Thm.~1]{arXiv2603}. Real formulations reproducing the predictions of Q-QIT also have a long history, going back to Stueckelberg~\cite{Stueckelberg} and subsequent multipartite extensions~\cite{Wootters}.

A natural goal is then to identify \emph{physical principles} that allow one to derive, or to select, the appropriate quantum theory over real Hilbert spaces, i.e., to choose between $T_1^{\mathds{R}}$ and $T_2^{\mathds{R}}$. Before stating the principle postulated in~\cite{arXiv2603}, we define two notions of independence, which play a central role in what follows.

\begin{definition}[Independent preparability of states]\label{Def:IndepPrep}
A bipartite state $\rho_{AB}$ is independently preparable if and only if it is a product state $\rho_{AB}=\rho_A\boxtimes\rho_B$, where $\boxtimes$ denotes the composition rule of the theory (e.g., the Kronecker tensor product $\otimes$ in Q-QIT, $T_1^{\mathds{R}}$ and $T_2^{\mathds{R}}$) and $\rho_A$, $\rho_B$ are valid (physical) states of the respective subsystems~\cite{VerstraeteCirac2003,erba2024composition}. Physically, these are the states that two separated parties can prepare by acting locally on their own subsystems, without communication, shared randomness, or other shared resources. 
\end{definition}

\begin{definition}[Operational independence]\label{Def:OpIndep}
A bipartite state $\rho_{AB}$ is operationally independent if and only if, for all local physical measurements $\{M_a\}_a$ on $A$ and $\{M_b\}_b$ on $B$, the joint outcome statistics factorise,
\[
p(a,b)=\Tr\big[(M_a\boxtimes M_b)\,\rho_{AB}\big]=p(a)\,p(b),
\]
where $p(a)$ and $p(b)$ denote the marginal distributions~\cite[Def.~2]{arXiv2603}.
\end{definition}

Both notions are operationally meaningful, but at different levels: independent preparability concerns whether a state can be produced by an independent preparation procedure, whereas operational independence refers only to the statistics of local measurements performed on the state. Every independently preparable state is operationally independent; whether the converse holds depends on the theory, as discussed below. In particular, the name ``operational independence'', taken from~\cite{arXiv2603}, should not be read as implying that independent preparation is not operational.

The authors of~\cite{arXiv2603} start from $T_1^{\mathds{R}}$ and argue that it should be modified into $T_2^{\mathds{R}}$ by enlarging the set of states regarded as independently preparable. Their argument can be summarized in the following Postulate, which they adopt to select the appropriate quantum theory over real Hilbert spaces:

\begin{postulate}[p.4 of~\cite{arXiv2603}, paraphrased]\label{Postulate:IndepPrepOpeIndep}
Independent preparation should be defined by operational independence (Definition~\ref{Def:OpIndep}), the only experimentally motivated notion of independence.
\end{postulate}

By construction, $T_2^{\mathds{R}}$ satisfies this Postulate, whereas $T_1^{\mathds{R}}$ does not, since in $T_1^{\mathds{R}}$ product-state independence is strictly stronger than operational independence~\cite[Prop.~3]{arXiv2603}. Hence, Postulate~\ref{Postulate:IndepPrepOpeIndep} selects $T_2^{\mathds{R}}$ as the appropriate quantum theory over real Hilbert spaces (or, equivalently, suggests that $T_1^{\mathds{R}}$ should be replaced by $T_2^{\mathds{R}}$). This theory reproduces the predictions of Q-QIT, and therefore cannot be ruled out by observing correlations allowed by Q-QIT.

Postulate~\ref{Postulate:IndepPrepOpeIndep} thus replaces Definition~\ref{Def:IndepPrep} by Definition~\ref{Def:OpIndep} as the definition of independent preparation. In $T_2^{\mathds{R}}$, this amounts to designating as independently preparable some states that are not of the product form $\rho_A\otimes\rho_B$, and hence do not satisfy Definition~\ref{Def:IndepPrep}. This is unavoidable. As we show in Appendix~\ref{App:LocalTomography}, within the finite-dimensional framework of generalised probabilistic theories, Definitions~\ref{Def:IndepPrep} and~\ref{Def:OpIndep} coincide if and only if the theory is locally tomographic\footnote{A theory is locally tomographic if every composite state is uniquely determined by the statistics of local measurements on its subsystems.}. Since $T_1^{\mathds{R}}$ is not locally tomographic~\cite{arXiv2603}, it contains operationally independent states that are not product states, and so does $T_2^{\mathds{R}}$, because local tomography depends only on the state and effect spaces, which $T_2^{\mathds{R}}$ shares with $T_1^{\mathds{R}}$. The question is therefore whether this redefinition of independent preparation is physically justified.

We examine whether the identification proposed in Postulate~\ref{Postulate:IndepPrepOpeIndep} can serve as a general physical principle governing independent preparation, including for fermionic systems. Our test case is \emph{fermionic quantum information theory} (F-QIT), which describes information encoded in the occupation of fermionic modes, i.e., in the presence or absence of identical fermions~\cite{banuls2009entanglement,d2014feynman,d2014fermionic,roos2017revealing,Vidal2021}. The differences between F-QIT and Q-QIT are discussed in Appendix~\ref{Sec:InfoTheories}. Like $T_1^{\mathds{R}}$, F-QIT is not locally tomographic~\cite{d2014fermionic,d2014feynman}. Our argument relies on the following claim:

\begin{claim*}
As fermions exist, any ``general physical principle'' should be shown valid within F-QIT.
\end{claim*}

In Sec.~\ref{F-QITCounterexample}, we provide a counterexample: a fermionic state that is operationally independent, and hence classified as independently preparable by Postulate~\ref{Postulate:IndepPrepOpeIndep}, although it cannot be prepared using local operations and shared randomness alone. Accepting our Claim, this shows that the Postulate cannot be regarded as a fully general physical principle, and that its use in selecting a quantum theory over real numbers requires further physical justification. This conclusion does not challenge the mathematical simulation result of~\cite{arXiv2603}.

In Appendix~\ref{App:OtherPostulates}, we also discuss two recent works~\cite{arXiv2503, arXiv2504} that reconstruct other real formulations of Q-QIT based on an alternative postulate (Postulate~\ref{Postulate:P4}), and give indications, based on information-theoretic considerations for indistinguishable particles, that its compatibility with F-QIT deserves further investigation.

More broadly, our comment highlights the importance of confronting proposed general physical principles with the full range of known physical theories. Historically, principles that appeared natural from a classical perspective were later shown to be incompatible with quantum phenomena. This is famously reflected in Einstein's remark that ``God does not play dice'', and more concretely illustrated by the assumptions underlying the Einstein--Podolsky--Rosen argument~\cite{EPR}, which were later shown by Bell's theorem~\cite{Bell} to be incompatible with quantum predictions. In a similar spirit, any proposed general physical principle should be carefully examined in light of frameworks such as F-QIT.

\section{Fermionic counterexample}\label{F-QITCounterexample}

As in the superselection-rule examples of Ref.~\cite{VerstraeteCirac2003}, the key point here is that the restricted class of physically allowed local operations may be unable to distinguish a locally prepared state from a non-locally prepared one.
Fermionic systems are subject to the \emph{parity superselection rule}, which forbids coherent superpositions between states of even and odd fermion number~\cite{Vidal2021}; this rule in particular prevents signalling due to fermionic anticommutation. As a consequence, both states and measurements are restricted.

We now briefly recall the notion of fermionic bits and the composition rule in F-QIT. As explained in Appendix~\ref{Sec:InfoTheories}, the basis states $\ket{0}$ and $\ket{1}$ represent the absence and presence of a fermionic particle, respectively, and the notions of locality and separability are defined with respect to modes~\cite{Vidal2021}. The composition rule is given by the antisymmetric wedge product $\wedge$. Formally, for basis states $\ket{n_A}_A$ and $\ket{n_B}_B$, with $n_A,n_B\in\{0,1\}$, the wedge product is defined as
\begin{equation}
    \ket{n_A n_B}_{AB}
    \equiv
    \ket{n_A}_A \wedge \ket{n_B}_B
    :=
    (f_A^\dagger)^{n_A}(f_B^\dagger)^{n_B}\ket{00},
\end{equation}
where $\ket{00}$ denotes the global vacuum state, and $A$ and $B$ label the fermionic modes.

The wedge product can be extended to operators by defining it on the operator basis elements $\ketbra{n_A}{n'_A}_{A}$ and $\ketbra{n_B}{n'_B}_{B}$, with $n_A,n'_A,n_B,n'_B\in\{0,1\}$, according to
\begin{equation}
   \ketbra{n_A}{n'_A}_{A}
   \wedge
   \ketbra{n_B}{n'_B}_{B}
   :=
   \big(
   \ket{n_A}_{A} \wedge \ket{n_B}_{B}
   \big)
   \left(
   \ket{n'_A}_{A} \wedge \ket{n'_B}_{B}
   \right)^\dagger.
\end{equation}

Now, consider two fermionic modes $A$ and $B$, and the Bell states of even and odd parity
\begin{equation}
\ket{\phi^+} = \frac{1}{\sqrt{2}}(\ket{00} + \ket{11}), 
\qquad
\ket{\psi^+} = \frac{1}{\sqrt{2}}(\ket{01} + \ket{10}).
\end{equation}

Define the mixed state
\begin{equation}
    \rho_{AB} = \tfrac{1}{2}\bigl(\ket{\phi^+}\!\bra{\phi^+} + \ket{\psi^+}\!\bra{\psi^+}\bigr).
\end{equation}
In the occupation-number basis $\{\ket{n_An_B}\}$, it reads
\begin{equation}\label{eq:rhoAB}
    \rho_{AB}
    =
    \frac{\id_{AB}}{4}
    +\frac{1}{4}\Big(\ketbra{00}{11}+\ketbra{11}{00}+\ketbra{01}{10}+\ketbra{10}{01}\Big).
\end{equation}
This state contains no coherences between the even and odd parity sectors. It is thus compatible with the parity superselection rule and is therefore a valid fermionic state.

As observed in~\cite[Sec.~IV.A]{Vidal2021}, not all operationally independent fermionic states (\cite[Def.~10]{Vidal2021}) can be realized as fermionic product states (i.e., states independently preparable in the sense of Definition~\ref{Def:IndepPrep}), and we now show that $\rho_{AB}$ is precisely such a state.

\paragraph{$\rho_{AB}$ is not independently preparable (by Definition~\ref{Def:IndepPrep}).}
In F-QIT, independently prepared systems correspond to product states of the form $\rho_A \wedge \rho_B$, where each local state must itself satisfy the parity superselection rule. For a single mode, the parity superselection rule forbids coherences between $\ket{0}$ and $\ket{1}$, so that every physical local state is diagonal in the occupation-number basis, $\rho_A=p\ketbra{0}{0}+(1-p)\ketbra{1}{1}$ with $p\in[0,1]$, and similarly for $\rho_B$.

Since the diagonal of $\rho_{AB}$ in Eq.~\eqref{eq:rhoAB} is uniform, its reduced states are $\rho_A=\rho_B=\id/2$. The only product state with these marginals is $\tfrac{\id}{2}\wedge\tfrac{\id}{2}=\tfrac{\id_{AB}}{4}\neq\rho_{AB}$. Hence, $\rho_{AB}$ is not independently preparable.

In fact, a stronger statement holds: $\rho_{AB}$ cannot be prepared by $A$ and $B$ using local operations and shared randomness alone. Indeed, wedge products of diagonal local states are diagonal, $\ketbra{n_A}{n_A}_A\wedge\ketbra{n_B}{n_B}_B=\ketbra{n_An_B}{n_An_B}$, and so are their convex combinations, whereas $\bra{00}\rho_{AB}\ket{11}=\tfrac14\neq0$.
Note that $\rho_{AB}=\tfrac12\big(\ketbra{++}{++}+\ketbra{--}{--}\big)$, with $\ket{\pm}=(\ket{0}\pm\ket{1})/\sqrt{2}$, is separable when viewed as a two-qubit state: it is the parity superselection rule, which excludes the local states $\ket{\pm}$, that prevents its local preparation.

\paragraph{$\rho_{AB}$ is operationally independent.}
In F-QIT, any allowed local effects $M_a$ on $A$ and $M_b$ on $B$ must commute with the local parity operator
\[
\Pi := \ket{0}\!\bra{0}-\ket{1}\!\bra{1},
\]
i.e. $[M_a,\Pi_A]=0$ and $[M_b,\Pi_B]=0$. For a single mode, this means that local effects are diagonal in the occupation-number basis,
\[
M_a=\sum_{n\in\{0,1\}} m_a(n)\,\ketbra{n}{n},
\qquad
M_b=\sum_{n\in\{0,1\}} m_b(n)\,\ketbra{n}{n},
\]
with $0\le m_a(n),m_b(n)\le1$. Consequently, $M_a\wedge M_b=\sum_{n_A,n_B}m_a(n_A)\,m_b(n_B)\ketbra{n_An_B}{n_An_B}$ only probes the diagonal of $\rho_{AB}$, which is uniform. Hence,
\[
p(a, b)
=
\Tr\!\big[(M_a\wedge M_b)\rho_{AB}\big]
=
\frac14\sum_{n_A,n_B}m_a(n_A)\,m_b(n_B)
=
p(a)\,p(b),
\]
where $p(a)=\Tr\!\big[(M_a\wedge\id_B)\rho_{AB}\big]=\tfrac12\sum_{n}m_a(n)$, and similarly $p(b)=\tfrac12\sum_{n}m_b(n)$.
Thus, $\rho_{AB}$ is operationally independent. Equivalently, $\rho_{AB}$ and $\id_{AB}/4$ cannot be distinguished by any physical local measurement. This indistinguishability is restricted to physical local measurements: the physical global effect $\ketbra{\phi^+}{\phi^+}$, which commutes with the global parity, distinguishes the two states, since
\[
\Tr\!\big[\ketbra{\phi^+}{\phi^+}\rho_{AB}\big]=\frac12,
\qquad
\Tr\!\big[\ketbra{\phi^+}{\phi^+}\tfrac{\id_{AB}}{4}\big]=\frac14.
\]

\paragraph{Consequences for Postulate~\ref{Postulate:IndepPrepOpeIndep}.}
We first observe that, in F-QIT, Definitions~\ref{Def:IndepPrep} and~\ref{Def:OpIndep} do not coincide: the state $\rho_{AB}$ is operationally independent, yet it is not independently preparable in the sense of Definition~\ref{Def:IndepPrep}. This is expected, since F-QIT is not locally tomographic (Appendix~\ref{App:LocalTomography}).

Postulate~\ref{Postulate:IndepPrepOpeIndep} resolves such a mismatch by redefining independent preparation through Definition~\ref{Def:OpIndep}, as $T_2^{\mathds{R}}$ does for $T_1^{\mathds{R}}$. Applied to F-QIT, it would regard $\rho_{AB}$ as independently preparable. However, which states separated parties can prepare using local operations and shared randomness alone is fixed by the local operations of the theory, which in this case must respect the parity superselection rule to avoid signalling~\cite{Vidal2021}. As shown above, every state obtainable in this way by two single-mode parties is diagonal in the occupation-number basis. Since $\rho_{AB}$ has nonzero off-diagonal entries, it cannot be prepared in this way (see also~\cite[Prop.~1]{VerstraeteCirac2003}).

Adopting Postulate~\ref{Postulate:IndepPrepOpeIndep} in F-QIT changes what independent sources are allowed to produce. Consider two independently operated sources with no shared resources, emitting modes $A$ and $B$, each with maximally mixed marginal, and a third party measuring the global effect $\ketbra{\phi^+}{\phi^+}$ on both modes. F-QIT predicts the outcome probability $\Tr\big[\ketbra{\phi^+}{\phi^+}\tfrac{\id_{AB}}{4}\big]=\tfrac14$, since the joint state produced by these independent preparations must be $\id_{AB}/4$. Under Postulate~\ref{Postulate:IndepPrepOpeIndep}, the value $\Tr\big[\ketbra{\phi^+}{\phi^+}\rho_{AB}\big]=\tfrac12$ would also be allowed. The Postulate therefore permits predictions incompatible with these preparation procedures in F-QIT.

One may still choose to call $\rho_{AB}$ ``independent'', or ``uncorrelated'' in the terminology of~\cite[Def.~10]{Vidal2021}, since it exhibits no correlations in physical local measurements.
These terms, however, do not imply that the state admits an independent preparation. Such a preparation is precisely the physical situation that independent sources are meant to describe in~\cite{arXiv2603}: systems prepared locally in separated laboratories without shared resources.

\paragraph{Generalisation.}
Our argument relies on the same mechanism as in~\cite{VerstraeteCirac2003}: in the presence of a superselection rule, all allowed local observables commute with the corresponding local symmetry, so states related by the associated twirling operation yield identical statistics for observables invariant under that symmetry.
In our setting, the relevant twirling is local twirling on the subsystems with respect to the local parity operator $\Pi$,
\[
\widetilde{\rho}_{AB}
:=
\frac{1}{4}\sum_{x,y\in\{0,1\}}
\big((\Pi_A)^x\wedge (\Pi_B)^y\big)\,\rho_{AB}\,\big((\Pi_A)^x\wedge (\Pi_B)^y\big),
\]
which removes all coherences between different local parities; here, $\widetilde{\rho}_{AB}=\id_{AB}/4$. Since physical local effects commute with $\Pi$, states related by this local twirling yield identical statistics for all physical local observables.
The same reasoning applies more broadly to theories with superselection rules, for instance to settings with charge, particle-number conservation, or anyonic information theories.

Our example is a realization of this mechanism in F-QIT, where our $\widetilde{\rho}_{AB}$ is $\mathcal{N}(\rho)$ in the notation of~\cite{VerstraeteCirac2003}: although Ref.~\cite{VerstraeteCirac2003} considers particle-number superselection for distinguishable subsystems, for single-mode parties the local parity projectors coincide with the local particle-number projectors defining $\mathcal{N}$.
Since $\rho_{AB} \neq \widetilde{\rho}_{AB}$, the fermionic analogue of~\cite[Prop.~1]{VerstraeteCirac2003}, which has the same proof, implies that $\rho_{AB}$ cannot be prepared by $A$ and $B$ even using local operations and classical communication (LOCC); in particular, it is not independently preparable.
On the other hand, $\rho_{AB}$ and $\widetilde{\rho}_{AB}$ cannot be distinguished on a single copy by physical local operations, even assisted by classical communication~\cite[Prop.~2]{VerstraeteCirac2003}, unless additional shared quantum resources are available.

\section{Conclusion}

Deriving the structure of composite systems and the notion of independent preparation from physically motivated principles is a compelling program. Independent preparability (Definition~\ref{Def:IndepPrep}) concerns the existence of a preparation by local operations without communication or shared resources, whereas operational independence (Definition~\ref{Def:OpIndep}) concerns local measurement statistics. Within the finite-dimensional framework of generalised probabilistic theories considered here, the two notions coincide if and only if the theory is locally tomographic. In F-QIT, a physically well-established
framework~\cite{banuls2009entanglement,d2014feynman,d2014fermionic,roos2017revealing,Vidal2021},
our counterexample state is operationally independent but cannot be
prepared using local operations and shared randomness alone.
For fixed marginal distributions, standard F-QIT permits a unique independently
prepared state, whereas Postulate~\ref{Postulate:IndepPrepOpeIndep}
also admits our counterexample state. These states are indistinguishable
by physical local measurements but distinguishable by a global
measurement. Adopting the postulate changes the global observable predictions for independent preparation procedures.
Therefore, Postulate~\ref{Postulate:IndepPrepOpeIndep} cannot be regarded as a fully general
physical principle.

Several recent works~\cite{arXiv2503,arXiv2504} follow a similar approach, proposing alternative real formulations to recover Q-QIT.
As we discuss in Appendix~\ref{App:OtherPostulates}, their proposed physically motivated formulation is not straightforward to reconcile with F-QIT and requires further justification.

More generally, we believe that the intuition behind these proposed postulates is rooted in some properties of Q-QIT, such as the absence of superselection rules on local operations and local tomography. These properties do not hold in F-QIT. This explains why, in F-QIT, Postulate~\ref{Postulate:IndepPrepOpeIndep} gives the wrong notion of independent preparation, and why the applicability of Postulate~\ref{Postulate:P4} to F-QIT requires further examination. From this perspective, it is not surprising that the real formulations
proposed in these works reproduce Q-QIT, since their underlying
assumptions are motivated by its structure. What requires further
justification is whether these assumptions can serve as general
physical principles, applicable also to F-QIT.

\section*{Acknowledgements}
We thank Dagmar Bru\ss, Timothée Hoffreumon and Anton Trushechkin for the constructive discussions we had.

\printbibliography

\appendix

\section{Information theories}\label{Sec:InfoTheories}

In this note, we distinguish qubit quantum information theory
(Q-QIT) and fermionic quantum information theory (F-QIT) as two
\emph{information-theoretic frameworks} within quantum theory.
Fermionic systems are, of course, quantum systems~\cite[Sec.~6]{volovich2002seven};
the distinction concerns the choice of physical systems used as
information carriers and how information is encoded in them.

By Q-QIT, we mean the standard formalism of quantum information,
as presented, for instance, in the textbook of Nielsen and
Chuang~\cite{NielsenChuang}. In this framework, information is
carried by distinguishable subsystems and may be encoded in
physical degrees of freedom such as spin, polarization, or
vibrational and electronic excitations. For example, with exactly
one electron localized at each site, the spin of each electron
encodes a qubit, with $\ket{0}$ and $\ket{1}$ corresponding to
spin up and spin down. Composite systems are described by the Kronecker tensor
products of Hilbert spaces; in fixed local bases, independently
prepared states are represented by Kronecker products of their
density matrices. The Q-QIT formalism itself imposes no
superselection restriction on local operations.

However, in the mode-based description (second quantization),
information is instead encoded in occupation numbers. We distinguish
the bosonic and fermionic cases. For bosons, in the absence of a
particle-number superselection rule, the resulting information-theoretic
framework is essentially equivalent to Q-QIT, albeit with
infinite-dimensional systems.

For identical fermions, the resulting framework is F-QIT.
Its subsystems are fermionic modes, which may be labelled by spatial
or spin degrees of freedom, or both. The states $\ket{0}$ and
$\ket{1}$ represent the absence and presence of a fermion in a given
mode, respectively. Subsystems compose through the antisymmetric wedge product $\wedge$,
defined in terms of fermionic operators satisfying the
canonical anticommutation relations~\cite{Vidal2021}. In the example above, the electron at each site is described using
two modes, $\uparrow$ and $\downarrow$: spin up corresponds to
$\ket{1}_{\uparrow}\wedge\ket{0}_{\downarrow}$ and spin down to
$\ket{0}_{\uparrow}\wedge\ket{1}_{\downarrow}$.

Several structural features distinguish F-QIT from Q-QIT.
First, physical states and operations must respect the parity
superselection rule, a consequence of fermionic anticommutation
together with the impossibility of signalling~\cite{d2014feynman,Vidal2021}.
Second, fermionic locality is not generally preserved under
encodings into qubits: standard mappings, such as the Jordan--Wigner
transformation, can map local fermionic operators to nonlocal qubit
operators. Exact fully local encodings are possible only in highly
restricted geometries~\cite{Guaita2025localityofqubit}.
Third, F-QIT is not locally tomographic~\cite{d2014fermionic,d2014feynman};
see Appendix~\ref{App:LocalTomography} for a formal definition.
Finally, the definition of entanglement must account for the
physical subsystems structure and superselection
constraints~\cite{wiseman2003entanglement,banuls2007entanglement}.

These distinctions motivate examining F-QIT as an information-theoretic
framework in its own right, rather than merely as a restricted subclass
of Q-QIT, when assessing proposed general physical principles.

\section{Operational independence coincides with independent preparability if and only if the theory is locally tomographic}\label{App:LocalTomography}

An information theory is said to be \emph{locally tomographic} if any state of a composite system is fully characterized by the statistics of local measurements performed on its subsystems. Q-QIT satisfies this property, whereas $T_1^{\mathds{R}}$ and F-QIT do not. Enlarging the preparation set as in $T_2^{\mathds{R}}$ does not, by itself, restore local tomography. A real formulation isomorphic to Q-QIT, such as RNQT~\cite[Def.~1]{arXiv2603} (introduced in~\cite{arXiv2504}), inherits local tomography, but it requires modifying the composition rule and is therefore a different construction from $T_2^{\mathds{R}}$. More precisely, in $T_1^{\mathds{R}}$ and F-QIT, there exist pairs of bipartite states $(\rho_{AB}, \sigma_{AB})$ that give identical statistics for every product of physical local effects, but that can nevertheless be distinguished by a global measurement on the joint system.

Throughout this appendix, a normalized state is independently preparable in the sense of Definition~\ref{Def:IndepPrep} precisely when it can be written as the bilinear composition $s_A\boxtimes s_B$ of two normalized physical local states. Operational independence is understood in the sense of Definition~\ref{Def:OpIndep}. With this convention, we show that, in finite-dimensional generalised probabilistic theories (GPTs), the two notions coincide if and only if the theory is locally tomographic. Consequently, since $T_1^{\mathds{R}}$ and F-QIT are not locally tomographic~\cite{arXiv2603,d2014feynman}, the two notions do not coincide in these theories, and adopting Postulate~\ref{Postulate:IndepPrepOpeIndep} in such a theory necessarily designates as independently preparable some states that are not of product form, as in $T_2^{\mathds{R}}$.

A single system in a GPT is described by a real finite-dimensional vector space \(A\), together with a convex set of states \(S_A \subset A\) and a set of effects \(E_A \subset A^*\), where \(A^*\) denotes the dual space of \(A\).

Allowed states are vectors \(\omega \in S_A\), with normalized states satisfying \(u^A(\omega)=1\), where \(u^A \in E_A\) is the unique unit effect. We use the standard convention that normalized physical states span the underlying vector space and that physical effects span its dual; the same convention applies to every composite system~\cite{baldijao2026tomographically}. The set \(E_A\) specifies the admissible effects. These are linear functionals \(e \in A^*\) such that, for all \(\omega \in S_A\), the probabilities \(e(\omega)\) are well-defined, i.e., \(0 \le e(\omega) \le 1\).

The composition of systems is described by a bilinear product \(\boxtimes\) acting on both states and effects. In particular, product states \(s_A \boxtimes s_B\) are valid states of the composite system, and independent measurements on independently prepared systems satisfy
\[
(e^A \boxtimes e^B)(s_A \boxtimes s_B)
= e^A(s_A)\, e^B(s_B),
\]
for all effects \(e^A \in E_A\) and \(e^B \in E_B\). The unit effect of the composite system is given by \(u^{AB} = u^A \boxtimes u^B\), where \(u^A \in E_A\) and \(u^B \in E_B\) are the respective unit effects.

We denote joint and marginal probabilities by
\[
P(a,b) := (e^A \boxtimes e^B)(s_{AB}), 
\quad
P(a) := e^A(s_A), 
\quad
P(b) := e^B(s_B),
\]
where the reduced (marginal) states \(s_A \in S_A\) and \(s_B \in S_B\) are defined by
\[
\begin{aligned}
e^A(s_A)&=(e^A\boxtimes u^B)(s_{AB})
&&\text{for all }e^A\in E_A,\\
e^B(s_B)&=(u^A\boxtimes e^B)(s_{AB})
&&\text{for all }e^B\in E_B.
\end{aligned}
\]

In particular, local effects separate states, i.e., if two states $\omega^A$ and $\nu^A$ satisfy
\[
e(\omega^A) = e(\nu^A) \quad \forall\, e \in E_A,
\]
then $\omega^A = \nu^A$.

Let us also recall the notion of local tomography, which is not a necessary property for all GPTs. A GPT is said to be \emph{locally tomographic} if global states are fully characterized by local measurements, i.e.,
\[
(e^A \boxtimes e^B)(s_{AB}) = (e^A \boxtimes e^B)(t_{AB})
\quad \forall\, e^A \in E_A,\; e^B \in E_B
\;\Rightarrow\;
s_{AB} = t_{AB}.
\]

\begin{proposition}[Local tomography implies the equivalence of operational independence and independent preparability]
\label{prop:LocalTomographyImpliesEquivalence}
Assume that the theory is locally tomographic, i.e., for all states \(s_{AB}, t_{AB} \in S_{AB}\),
\[
(e^A \boxtimes e^B)(s_{AB}) = (e^A \boxtimes e^B)(t_{AB})
\quad \forall\, e^A \in E_A,\; e^B \in E_B
\;\Rightarrow\;
s_{AB} = t_{AB}.
\]
Then, for any normalized bipartite state \(s_{AB}\),
\[
s_{AB} = s_A \boxtimes s_B
\quad \Longleftrightarrow \quad
(e^A \boxtimes e^B)(s_{AB}) = e^A(s_A)\, e^B(s_B)
\]
for all effects \(e^A \in E_A\) and \(e^B \in E_B\).
\end{proposition}

\begin{proof}
\textbf{($\Rightarrow$)}  
Let \(s_{AB} = s_A \boxtimes s_B\). Then, by the defining property of composite systems in GPTs,
\[
(e^A \boxtimes e^B)(s_{AB})
= (e^A \boxtimes e^B)(s_A \boxtimes s_B)
= e^A(s_A)\, e^B(s_B).
\]

\medskip

\textbf{($\Leftarrow$)}  
Assume that for all effects \(e^A \in E_A\) and \(e^B \in E_B\),
\[
(e^A \boxtimes e^B)(s_{AB}) = e^A(s_A)\, e^B(s_B),
\]
where \(s_A\) and \(s_B\) are the marginal states of \(s_{AB}\).

Then, for all \(e^A \in E_A, e^B \in E_B\),
\[
(e^A \boxtimes e^B)(s_{AB})
= e^A(s_A)\, e^B(s_B)
= (e^A \boxtimes e^B)(s_A \boxtimes s_B),
\]
where the last equality again follows from the defining property of composite systems in GPTs.

Hence,
\[
(e^A \boxtimes e^B)(s_{AB})
=
(e^A \boxtimes e^B)(s_A \boxtimes s_B)
\quad \forall\, e^A \in E_A,\, e^B \in E_B.
\]

By local tomography, this implies
\[
s_{AB} = s_A \boxtimes s_B.
\]
\end{proof}

\begin{proposition}[Equivalence of operational independence and independent preparability implies local tomography]
\label{prop:EquivalenceImpliesLocalTomography}
Assume that in a GPT the equivalence of operational independence and independent preparability holds, i.e., for any normalized bipartite state \(s_{AB}\),
\[
s_{AB} = s_A \boxtimes s_B
\quad \Longleftrightarrow \quad
(e^A \boxtimes e^B)(s_{AB}) = e^A(s_A)\, e^B(s_B)
\]
for all effects \(e^A \in E_A\) and \(e^B \in E_B\).
Then the theory is locally tomographic.
\end{proposition}

\begin{proof}
Let $V_{AB}$ denote the vector space of the composite system. By Lemmas~1 and~2 of~\cite{baldijao2026tomographically},
\[
V_{AB}=L_{AB}\oplus H_S,
\qquad
L_{AB}:=\operatorname{Span}\bigl\{a\boxtimes b\mid a\in S_A,\ b\in S_B\bigr\},
\]
where $H_S$ is the holistic subspace. Every $h\in H_S$ is invisible to product effects:
\[
(e^A\boxtimes e^B)(h)=0
\quad\forall\,e^A\in E_A,\ e^B\in E_B.
\]
The theory is locally tomographic if and only if $H_S=\{0\}$.

Choose bases of normalized physical local states $\{a_i\}_{i=1}^{m}$ and $\{b_j\}_{j=1}^{n}$. Their products form a basis of $L_{AB}$. Define
\[
s_A:=\frac1m\sum_{i=1}^{m}a_i,
\qquad
s_B:=\frac1n\sum_{j=1}^{n}b_j,
\qquad
s:=s_A\boxtimes s_B
=\frac1{mn}\sum_{i,j}a_i\boxtimes b_j.
\]
Let $\omega\in S_{AB}$ be an arbitrary normalized physical state, and write
\[
\omega=v+h,
\qquad
v=\sum_{i,j}c_{ij}\,a_i\boxtimes b_j\in L_{AB},
\qquad h\in H_S.
\]
The coefficients $c_{ij}$ are real and need not be nonnegative. Since $u^{AB}=u^A\boxtimes u^B$ annihilates $h$, normalization gives
\[
1=u^{AB}(\omega)=u^{AB}(v)=\sum_{i,j}c_{ij}.
\]

To construct a physical state with product local statistics, let
\[
C:=\max_{i,j}c_{ij}\geq\frac1{mn},
\qquad
t:=\frac{1}{2mnC}.
\]
Then $0<t\leq\tfrac12$ and
\[
\frac1{mn}-tc_{ij}\geq\frac1{2mn}>0
\quad\text{for all }i,j.
\]
Consequently,
\[
\tau:=\frac{s-tv}{1-t}
=\sum_{i,j}\frac{\frac1{mn}-tc_{ij}}{1-t}\,a_i\boxtimes b_j
\]
is a normalized physical state: its coefficients are nonnegative and sum to one. By convexity, the state
\[
\widetilde{s}:=t\omega+(1-t)\tau=s+th
\]
is also physical.

Since $h$ is invisible to all product effects,
\begin{equation}\label{eq:proofofprop2}
(e^A\boxtimes e^B)(\widetilde{s})
=(e^A\boxtimes e^B)(s)
=e^A(s_A)\,e^B(s_B)
\end{equation}
for all $e^A\in E_A$ and $e^B\in E_B$. Taking $e^B=u^B$ and then $e^A=u^A$, and using that local effects separate local states, shows that the marginals of $\widetilde{s}$ are $s_A$ and $s_B$. Thus, Eq.~\eqref{eq:proofofprop2} establishes that $\widetilde{s}$ is operationally independent. By the hypothesis of the proposition,
\[
\widetilde{s}=s_A\boxtimes s_B=s.
\]
Together with $\widetilde{s}=s+th$ and $t>0$, this implies $h=0$.

Therefore, every normalized physical state belongs to $L_{AB}$. Since these states span $V_{AB}$, we have $V_{AB}=L_{AB}$ and hence $H_S=\{0\}$. The theory is locally tomographic.
\end{proof}

\section{On subsystem compositions and locality for indistinguishable particles}\label{App:OtherPostulates}
Two recent works~\cite{arXiv2503,arXiv2504} propose alternative notions of subsystem composition and locality, under which Q-QIT can be reconstructed from real Hilbert spaces.
In this appendix, we simply note that the applicability of these axiomatic notions is not immediate in the context of indistinguishable particles, an important class of physical information carriers.
In the same spirit as our Claim, separate justification for F-QIT is needed.

For concreteness, we consider the following adapted version of
Postulate~(P4) from Ref.~\cite{arXiv2503}.

\begin{postulate}\label{Postulate:P4}
Let subsystems $A$ and $B$ be associated with Hilbert spaces $\mathcal H_A$ and $\mathcal H_B$, and let the composite system be represented on a Hilbert space $\mathcal H_{AB}$. Suppose that independently prepared states are embedded through a bilinear map
\[
\boxtimes : \mathcal H_A \times \mathcal H_B \to \mathcal H_{AB},
\qquad
(\ket{\psi}_A,\ket{\phi}_B)\mapsto \ket{\psi}_A \boxtimes \ket{\phi}_B.
\]
For each physical local operator $O_B \in \mathcal L(\mathcal H_B)$, there is a corresponding operator $(\id_A \boxtimes O_B)$ on $\mathcal H_{AB}$ such that
\begin{equation}\label{Eq:PostulateP4}
(\id_A \boxtimes O_B)(\ket{\psi}_A \boxtimes \ket{\phi}_B)
=
\ket{\psi}_A \boxtimes (O_B\ket{\phi}_B),
\end{equation}
and similarly for every physical local operator on $A$.
\end{postulate}

This formulation is natural in the usual framework of distinguishable subsystems\footnote{If the composition map $\boxtimes$ is bilinear, then by the universal property of the tensor product it induces a unique linear map from $\mathcal H_A \otimes \mathcal H_B$ to $\mathcal H_{AB}$. In this sense, the formalism remains closely tied to an underlying tensor-product structure. The same discussion extends to multipartite systems with multilinear $\boxtimes$.}.
However, for indistinguishable particles the status of such a subsystem picture is more delicate information-theoretically~\cite{wiseman2003entanglement,benatti2021entanglement,garrison2022quantum}.

Indeed, in the first-quantized, particle-based description of identical particles, particle labels do not define physical subsystems in the same way as for distinguishable systems.
For indistinguishable particles, one cannot in general interpret $O_A \otimes \id$ and $\id \otimes O_B$ as physical observables acting on “the first” or “the second” particle.
Since particle labels have no physical meaning for indistinguishable particles, physically admissible observables must be invariant under particle exchange. 
Accordingly, an observable describing a local joint physical measurement on two identical particles need not factorize into a product of two independently physical one-particle observables~\cite[Sec.~2]{benatti2021entanglement}.

For this reason, in information-theoretic treatments of indistinguishable particles (such as F-QIT), the more natural framework is the mode-based, second-quantized one, where locality is associated with subsets of modes rather than labelled particles~\cite{Vidal2021}.
From this perspective, the applicability of Postulate~\ref{Postulate:P4} is not automatic: it should be formulated directly in the relevant indistinguishable-particle framework, that is second quantization, where its justification is not trivial.

\end{document}